\documentclass[aps,pra,showpacs,twocolumn,amsmath,superscriptaddress]{revtex4-1}
\usepackage{color}
\usepackage{graphicx,epsfig,subfigure,dsfont,amssymb,amsmath,amsthm,amsfonts,amsbsy,mathrsfs,amscd}
\usepackage{amsmath}
\usepackage{bm}
\usepackage{bbm}
\usepackage{contour}
\usepackage{latexsym}
\usepackage{amssymb} 
\usepackage{amsthm} 
\usepackage{times} 
\usepackage{graphicx}

\newcommand{\aop}{Ann. Phys.~}

\newcommand{\jmp}{J. Math. Phys.~}
\newcommand{\jpa}{J. Phys. A: Math. Theor.~}

\newcommand{\njp}{New. J. Phys.~}

\newcommand{\tinyspace}{\mspace{1mu}}

\newcommand{\abs}[1]{\left\lvert\tinyspace #1 \tinyspace\right\rvert}

\newcommand{\norm}[1]{\left\lVert\tinyspace #1 \tinyspace\right\rVert}

\renewcommand{\det}{\operatorname{det}}
\renewcommand{\t}{{\scriptscriptstyle\mathsf{T}}}

\newcommand{\setft}[1]{\mathrm{#1}}

\newcommand{\density}[1]{\setft{D}\left(#1\right)}

\def\dif{\mathrm{d}}

\def\complex{\mathbb{C}}
\def\real{\mathbb{R}}

\def\I{\mathbbm{1}}

\def\dif{\mathrm{d}}
\def\1{\mathbf{1}}

\newcommand{\out}[2]{| #1\rangle\langle #2 |}
\newcommand{\Inner}[2]{\left\langle #1 , #2\right\rangle}
\newcommand{\Innerm}[3]{\left\langle #1 \left| #2 \right| #3 \right\rangle}

\newcommand{\Pa}[1]{\left(#1\right)}

\newcommand{\Br}[1]{\left[#1\right]}
\newcommand{\set}[1]{\{#1\}}
\newcommand{\Set}[1]{\left\{#1\right\}}

\newcommand{\ket}[1]{|#1\rangle}

\DeclareMathOperator{\trace}{Tr}

\newcommand{\Ptr}[2]{\trace_{#1}\Pa{#2}}

\newcommand{\Tr}[1]{\Ptr{}{#1}}

\newcommand{\Abs}[1]{\left|\tinyspace#1\tinyspace\right|}

\def\cB{\mathcal{B}}\def\cE{\mathcal{E}}
\def\cF{\mathcal{F}}\def\cG{\mathcal{G}}

\def\cR{\mathcal{R}}

\def\bsF{\boldsymbol{F}}

\def\bsP{\boldsymbol{P}}\def\bsR{\boldsymbol{R}}

\def\bsa{\boldsymbol{a}}\def\bsb{\boldsymbol{b}}\def\bsc{\boldsymbol{c}}\def\bse{\boldsymbol{e}}

\def\bsp{\boldsymbol{p}}\def\bsr{\boldsymbol{r}}
\def\bsu{\boldsymbol{u}}\def\bsv{\boldsymbol{v}}

\newtheorem{thrm}{Theorem}
\newtheorem{lem}{Lemma}
\newtheorem{prop}{Proposition}
\newtheorem{cor}{Corollary}
\theoremstyle{definition}
\newtheorem{definition}{Definition}

\begin{document}


\title{Geometry of sets of Bargmann invariants}


\author{Lin Zhang}
\email{godyalin@163.com}
\author{Bing Xie}
\email{xiebingjiangxi2023@163.com} \affiliation{School of Science,
Hangzhou Dianzi University, Hangzhou 310018, China}
\author{Bo Li}
\email{libobeijing2008@163.com} \affiliation{School of Computer and
Computing Science, Hangzhou City University, Hangzhou 310015, China}

\date{\today}

\begin{abstract}
Certain unitary-invariants, known as Bargmann invariants or
multivariate traces of quantum states, have recently gained
attention due to their applications in quantum information theory.
However, determining the boundaries of sets of Bargmann invariants
remains a theoretical challenge. In this study, we address the
problem by developing a unified, dimension-independent formulation
that characterizes the sets of the 3rd and 4th Bargmann invariants.
In particular, our result for the set of 4th Bargmann invariants
confirms the conjecture given by Fernandes \emph{et al.} [\prl
\href{https://doi.org/10.1103/PhysRevLett.133.190201}{\textbf{133},
190201 (2024)}]. Based on the obtained results, we conjecture that
the unified, dimension-independent formulation of the boundaries for
sets of 3rd-order and 4th-order Bargmann invariants may extend to
the general case of the $n$th-order Bargmann invariants. These
results deepen our understanding of the fundamental physical limits
within quantum mechanics and pave the way for novel applications of
Bargmann invariants in quantum information processing and related
fields.
\end{abstract}

\maketitle

\section{Introduction}

Quantum mechanics provides a fundamental framework for describing
the behavior of physical systems, traditionally applied to
microscopic scales. A common characteristic of quantum mechanics is
its reliance on the formalism of complex Hilbert spaces, which
constitutes the mathematical foundation of the theory. However, this
feature is neither unique nor exclusive to quantum theory, as other
physical theories, including classical electromagnetism, also employ
complex numbers in their formulations. Furthermore, the use of
complex numbers is not an indispensable aspect of quantum mechanics,
as there exist real-valued formulations capable of reproducing all
quantum theoretical predictions. This distinction is particularly
relevant in light of recent advancements, such as the work by Renou
\emph{et al.} \cite{Renou2021}, which established a no-go theorem
falsifying a broad class of real quantum theories. Nevertheless,
their results do not universally exclude all real formulations, as
certain theories evade the constraints of their no-go theorem. These
insights underscore the nuanced relationship between mathematical
representations and physical predictions in quantum theory,
highlighting the need for careful consideration of foundational
assumptions.

The complex mathematical structure is associated with various
physical phenomena, including geometric phases
\cite{Berry1984,Aharonov1987,Mukunda1993,Chruscinski2004}, quantum
coherence \cite{Baumgratz2014}, and interference effects
\cite{Pati1995}, among others. Recently, several studies
\cite{Renou2021,Hickey1,Wu1,Wu17} have sparked interest in
investigating a particular type of quantum coherence known as
imaginarity. In studying such phenomena, Bargmann invariants
(proposed originally in \cite{Bargmann1964}, and developed in
\cite{Chien2016}), as a mathematical tool, provides a framework to
analyze the properties of quantum states that remain invariant under
unitary transformations. Bargmann invariants play a important role
in distinguishing quantum states and characterizing their
symmetries. They are gauge-independent quantities directly related
to the geometric phase acquired by quantum systems undergoing cyclic
evolution \cite{Simon1993}. Their connections to Wigner rotations
and null-phase curves further underscore their significance
\cite{Mukunda2003}. Additionally, Bargmann invariants are firstly
linked to Kirkwood-Dirac quasi-probability distributions
\cite{Dirac1945, Kirkwood1933, Bamber2014,DavidSchmid2024} by Wagner
\emph{et. al.} in \cite{Wagner2024}, providing insights into quantum
contextuality \cite{Budroni2022}, which is connected to Bargmann
invariants in Refs. \cite{Wagner2024pra,Giordani2023}, and
measurement-induced disturbances
\cite{Bievre2021,Bievre2023,Budiyono2023}. Their invariance
properties enable a deeper understanding of quantum state structures
and associated symmetries \cite{Chien2016,Galvao2020,Oszmaniec2024}.
Bargmann invariants are also known as multivariate traces, enabling
efficient estimation using constant-depth quantum circuits
\cite{Quek2024}, compatible with near-term quantum computing
architectures.

Although the role of Bargmann invariants within the formal framework
of quantum resource theories has not yet been fully explored,
several studies have already established connections between the
values of these invariants and various quantum resources. Recently,
the work of Fernandes \emph{et al.} \cite{Fernandes2024}, among
others, has initiated investigations into how these invariants can
be utilized to witness quantum imaginarity \cite{Wu1, Wu17}---a
property that has been independently recognized as a valuable
quantum resource within the framework of resource theories. This
perspective is gaining recognition as a relevant and promising
direction for future research.

A central challenge in this context is determining the boundaries
\cite{Fraser2023} of sets of Bargmann invariants of arbitrary
orders. These boundaries physically define meaningful constraints
imposed by quantum mechanics, and reveal intricate relationships
between Bargmann invariants and quantum imaginarity
\cite{Fernandes2024}. Determining such boundaries is a problem of
relevance to both quantum foundations and applications. Notably,
Bargmann invariants have been linked to quantum resources, such as
imaginarity \cite{Fernandes2024}, and indistinguishability
\cite{Shchesnovich2015,Menssen2017,Brod2019,Jones2020,Jones2023,Rodari2024a,Rodari2024b},
and to their use in benchmarking and certifying nonclassical
resources of quantum devices
\cite{Giordani2020,Giordani2021,Giordani2023,Hinsche2024}.
Researchers have proposed basis-independent witnesses for
imaginarity, constructed from unitary-invariant properties of
quantum states \cite{Fernandes2024}. While these studies fully
characterized invariant values for three mixed states, only partial
results exist for when considering four states or higher, leaving
several unresolved conjectures that motivate further investigations
of the boundaries of the sets of all Bargmann invariants.

The paper is organized as follows. In Section~\ref{sect:pre}, we
review the fundamental concepts of Bargmann invariants and establish
key definitions related to numerical ranges and envelope of a family
of curves. This includes foundational results that underpin our
subsequent analysis. Section~\ref{sect:mainre} presents our core
contributions: (1) a unified, dimension-independent characterization
of 3rd- and 4th-order Bargmann invariants, and (2) a rigorous proof
of Fernandes \emph{et al.}'s conjecture for the 4th-order case,
resolving an open question in the field. Section~\ref{sect:evidence}
provides qualitative arguments supporting the conjectured
generalization of our framework to $n$th-order Bargmann invariants.
We conclude with a summary of our results and outline future
research directions, including the use of Bargmann invariants for
entanglement detection---a direction we explore in depth in
companion work \cite{Lin2024}.

\section{Preliminaries}
\label{sect:pre}

For a tuple of quantum states $\Psi=(\rho_1,\ldots,\rho_n)$, where
each state $\rho_k$ acts on the underlying Hilbert space
$\complex^d$, the \emph{Bargmann invariant} associated with $\Psi$
is defined as $\Delta_n(\Psi)=\Tr{\rho_1\cdots\rho_n}$. Due to the
spectral decomposition of density operators,
$\rho_k=\sum^d_{j=1}\lambda_{k,j}\out{\psi_{k,j}}{\psi_{k,j}}$,
where $k\in\set{1,2,\ldots,n}$ and $\lambda_{k,j}$'s are eigenvalues
of $\rho_k$, then we see that
\begin{eqnarray*}
&&\Delta_n(\Psi) = \Tr{\rho_1\cdots\rho_n} \\
&&=\sum^d_{j_1,\ldots,j_n=1}\lambda_{j_1,\ldots,j_n}\Tr{\psi_{1,j_1}\cdots\psi_{n,j_n}},
\end{eqnarray*}
where $\psi_{k,j_k}\equiv\out{\psi_{k,j_k}}{\psi_{k,j_k}}$ and
$\lambda_{j_1,\ldots,j_n}=\prod^n_{k=1}\lambda_{k,j_k}$ is clearly a
multivariate probability distribution on $\set{1,\ldots,d}^n$. In
view of this relationship between the sets of Bargmann invariants
defined by pure states and mixed states, we introduce different
notations for them.
\begin{definition}[Bargmann invariants for quantum states]
We denote the set of possible values of $n$th-order Bargmann
invariants for a tuple $\Psi=(\rho_1,\rho_2,\ldots,\rho_n)$ of
quantum states $\rho_k$, where $k=1,\ldots,n$, in $\complex^d(d
\geqslant 2)$, by $\cB_n(d):=\Set{\Delta_n(\Psi): \rho_k \text{ acts
on }\complex^d,k=1,\ldots,n}$. In particular, we denote by
\begin{enumerate}
\item[(i)] $\cB^\circ_n(d)$: the set of $n$th-order Bargmann
invariants defined by only pure states.
\item[(ii)] $\cB^\bullet_n(d)$: the set of $n$th-order Bargmann
invariants defined by only mixed states.
\end{enumerate}
\end{definition}
Note that $\cB^\bullet_n(d)=\mathrm{ConvexHull}(\cB^\circ_n(d))$ by
definition. Thus it would be the case
$\cB^\bullet_n(d)=\cB^\circ_n(d)$ whenever $\cB^\circ_n(d)$ is
convex (it is indeed the case when $n\in\set{3,4}$ from
Theorem~\ref{th:one}). We firstly need figure out the properties of
$\cB^\circ_n(d)$. Following \cite{Fernandes2024}, the set of
quantum-realizable values for $n$th-order Bargmann invariants is
given by $\cB^\circ_n := \Set{\Delta \in \complex : \exists \Psi,
\Delta = \Delta_n(\Psi)}$, where $\Delta_n(\Psi):=
\Inner{\psi_1}{\psi_2} \cdots \Inner{\psi_{n-1}}{\psi_n}
\Inner{\psi_n}{\psi_1}$, defined for a tuple of pure states
$\Psi=(\out{\psi_1}{\psi_1},\ldots,\out{\psi_n}{\psi_n})$ in
$\complex^d$ for some $d \geqslant 2$. With these notations, we
obtain $\cB^\circ_n = \cup_{d \geqslant 2} \cB^\circ_n(d)$. The
properties of $\cB^\circ_n$ are as follows:
\begin{itemize}
    \item $\cB^\circ_n$ is a \emph{closed} and \emph{connected} set because it is the continuous image of Cartesian product of pure state spaces.
    \item $\cB^\circ_n$ is symmetric with respect to the real axis, as
    $\Tr{\psi_n \cdots \psi_2\psi_1} = \overline{\Tr{\psi_1\psi_2 \cdots \psi_n}}$.
    \item $\cB^\circ_n \subset \Set{z \in \complex : \abs{z} \leqslant 1}$ due to the fact that $\abs{\Tr{\rho_1\rho_2\cdots\rho_n}}\leqslant1$.
\end{itemize}
We see from the convexity of $\cB^\circ_3$ \cite{Fernandes2024} that
$\cB^\circ_3 = \mathrm{ConvexHull}(\partial \cB^\circ_3)$, where
$\partial \cB^\circ_3$ denotes the boundary of $\cB^\circ_3$.
Similarly, $\partial \cB^\circ_n$ indicates the boundary of
$\cB^\circ_n$ for a general natural number $n$. The notation
$\partial\cB^\circ_n(d)$ has a similar meaning. This suggests that
determining $\partial \cB^\circ_n$ is crucial for characterizing
$\cB^\circ_n$ and then establishing its convexity. In the following,
we present polar equations of boundaries $\partial\cB^\circ_n(d)$
and show that $\cB^\circ_n(d)=\cB^\circ_n(2)=\cB^\circ_n$ is a
convex set in $\complex$, where $n\in\set{3,4}$.

In the proof of the main result (Theorem~\ref{th:one}), we will make
use of the following notions: the numerical range of a complex
square matrix and envelope of a family of plane curves. Let me
explain it now.

\subsection{Numerical range}

\begin{definition}[Numerical range of a complex square matrix]
The numerical range (NR) of a $d\times d$ complex matrix $\bsF$ is
defied by
$$
W_d(\bsF):=\Set{\Innerm{\psi}{\bsF}{\psi}:\ket{\psi}\in\complex^d,\Inner{\psi}{\psi}=1}.
$$
\end{definition}
The notion of NR was introduced by Toeplitz \cite{Toeplitz1918} in
1918. In 1919, Hausdorff showed that $W_d(\bsF)\subset\complex$ is a
compact and convex set \cite{Hausdorff1919}. Based on the main
result in \cite{Chien2001}, we immediately obtain the following key
result concerning the NR of a rank-one complex square matrix.
\begin{prop}\label{prop:chien}
If a $d\times d$ complex matrix $\bsF$ is given by
$\bsF=\bsu\bsv^\dagger$, where $\bsu,\bsv\in\complex^d$. Then
$W_d(\bsF)$ is an elliptical disk in $\complex$ with foci $0$ and
$\Inner{\bsv}{\bsu}$, and minor axis
$(\Inner{\bsu}{\bsu}\Inner{\bsv}{\bsv}-\abs{\Inner{\bsu}{\bsv}}^2)^\frac12$.
\end{prop}
Let $2c=\abs{\Inner{\bsv}{\bsu}}$ and
$2b=(\Inner{\bsu}{\bsu}\Inner{\bsv}{\bsv}-\abs{\Inner{\bsu}{\bsv}}^2)^\frac12$.
Then the major axis is given by
$2a=\sqrt{(2b)^2+(2c)^2}=(\Inner{\bsu}{\bsu}\Inner{\bsv}{\bsv})^\frac12$.
Note that on the complex plane, the sum of distances from a moving
point to foci $0$ and $\Inner{\bsv}{\bsu}$ is constrained by
$\abs{z}+\abs{z-\Inner{\bsv}{\bsu}}\leqslant
2a=\norm{\bsu}\norm{\bsv}$, where
$\norm{\cdot}=\sqrt{\Inner{\cdot}{\cdot}}$. Therefore $z\in
W_d(\bsu\bsv^\dagger)$ if and only if
$\abs{z}+\abs{z-\Inner{\bsv}{\bsu}}\leqslant\norm{\bsu}\norm{\bsv}$
for $z\in\complex$. In particular, $W_d(\bsu\bsv^\dagger)$ is a
circle $\abs{z}\leqslant\frac12\norm{\bsu}\norm{\bsv}$ for
$\Inner{\bsu}{\bsv}=0$.

Now for an $n$-tuple $\Psi=(\psi_1,\ldots,\psi_n)$, where
$n\geqslant3$, of pure qudit states, when
$\psi_n\equiv\out{\psi_n}{\psi_n}$ varies and
$\psi_1,\ldots,\psi_{n-1}$ are fixed temporarily, then the
$n$th-order Bargmann invariant of this tuple $\Psi$ is given by
$\Tr{\psi_1\cdots\psi_{n-1}\psi_n} =
\Innerm{\psi_n}{\psi_1\psi_2\cdots\psi_{n-1}}{\psi_n}$. Based on
this observation, we deduce that the NR
$W_d\Pa{\psi_1\psi_2\cdots\psi_{n-1}}$ of
$\psi_1\psi_2\cdots\psi_{n-1}$ corresponds precisely to the values
of $n$th-order Bargmann invariants obtained by varying $\psi_n$
while keeping $\psi_k$ (for $k=1,\ldots,n-1$) fixed.

Next, Lemma~\ref{lem:1} indicates that the set of all $n$th-order
Bargmann invariants defined by pure qudit states can be represented
by the union of a family of scaling and rotated elliptical disk.
(Throughout the whole paper, the term ``elliptical disk" includes
``degenerate ellipses," that is, line segments and points.)

\begin{lem}\label{lem:1}
For any pure states $(\psi_1,\ldots,\psi_{n-1})$ in $\complex^d$,
where $\psi_k\equiv\out{\psi_k}{\psi_k}$, the numerical range of
$\psi_1\cdots\psi_{n-1}$ is given by
$$
W_d(\psi_1\cdots\psi_{n-1}) =
\Big(\prod^{n-2}_{j=1}\Inner{\psi_j}{\psi_{j+1}}\Big)W_d(\out{\psi_1}{\psi_{n-1}}).
$$
Moreover, if $\prod^{n-2}_{j=1}\Inner{\psi_j}{\psi_{j+1}}=0$, then
$W_d(\psi_1\cdots\psi_{n-1}) =\set{0}$; or else, if
$\prod^{n-2}_{j=1}\Inner{\psi_j}{\psi_{j+1}}\neq0$, then
$W_d(\psi_1\cdots\psi_{n-1})$ can be reduced to the following
elliptical disk: {\scriptsize$$
\Set{z\in\complex:\abs{z}+\abs{z-\Tr{\psi_1\cdots\psi_{n-1}}}\leqslant\Abs{\prod^{n-2}_{j=1}\Inner{\psi_j}{\psi_{j+1}}}},
$$}
denoted by $\cE_{\psi_1,\ldots,\psi_{n-1}}$. Furthermore, we have
\begin{eqnarray}\label{eq:1}
\cB^\circ_n(d)
=\bigcup_{\psi_1,\ldots,\psi_{n-1}}\cE_{\psi_1,\ldots,\psi_{n-1}},
\end{eqnarray}
where all $\psi_k$'s run over the set of pure states in
$\complex^d$.
\end{lem}

\begin{proof}
For any tuple of pure states $(\psi_1,\ldots,\psi_{n-1})$ in
$\complex^d$, the product is of rank-one, i.e.,
$\psi_1\cdots\psi_{n-1}=\Big(\prod^{n-2}_{j=1}\Inner{\psi_j}{\psi_{j+1}}\Big)\out{\psi_1}{\psi_{n-1}}$.
This implies that
\begin{eqnarray*}
W_d(\psi_1\cdots\psi_{n-1})
=\Big(\prod^{n-2}_{j=1}\Inner{\psi_j}{\psi_{j+1}}\Big)W_d(\out{\psi_1}{\psi_{n-1}}).
\end{eqnarray*}
From Proposition~\ref{prop:chien}, we see that
$W_d(\out{\psi_1}{\psi_{n-1}})=\Set{z\in\complex:\abs{z}+\abs{z-\Inner{\psi_{n-1}}{\psi_1}}\leqslant1}$.
The last identity follows immediately by scaling and rotating the
elliptical disk $W_d(\out{\psi_1}{\psi_{n-1}})$ via the factor
$\prod^{n-2}_{j=1}\Inner{\psi_j}{\psi_{j+1}}$. We obtain the desired
conclusion.
\end{proof}
From Eq.~\eqref{eq:1} in Lemma~\ref{lem:1}, we see that
$\cB^\circ_n(d)$ can be represented as a family of elliptical discs.
If we can parametrize this family of ellipses
$\partial\cE_{\psi_1,\ldots,\psi_{n-1}}$, the boundaries of the
elliptical disks, we are able to calculate the boundary
$\partial\cB^\circ_n(d)$, as the envelope of this family of curves.
Next, we introduce the notion of envelope of a family of curves.

\subsection{Envelope of a family of plane curves}

The envelope of a family of curves is a curve that is tangent to
each member of the family at some point, and these points of
tangency together form the envelope itself. In other words, the
envelope is a curve that ``wraps around" or ``bounds" the entire
family of curves, touching each curve in the family at least once.
Mathematically, if the family of curves is defined by a parameter,
the envelope can often be found by solving a system of equations
derived from the family's defining equations and their derivatives
with respect to the parameter \cite{Bickel2020}. For instance, the
elliptical range theorem for the numerical range of a $2\times2$
complex matrix can be derived by the envelope method. We often
compute boundaries by determining envelopes and appending a finite
number of additional points. Let us recall some notion about
envelope.
\begin{definition}[Envelope, \cite{Bickel2020}]
Let $\cF$ denote a family of curves in the $xy$ plane (rectangular
coordinate system). We assume that $\cF$ is a family of curves given
by $F(x,y,t)=0$, where $F$ is smooth and $t$ lies in an open
interval. The envelope of $\cF$ is the set of points $(x,y)$ so that
there is a value of $t$ with both $F(x,y,t)=0$ and
$\partial_tF(x,y,t)=0$.
\end{definition}

The envelope problem can be extended to multiple parameters
\cite{Hull2020}. For example, for a family of plane curves, defined
by $f(x,y,t_1,t_2)=0$, where $t_1$ and $t_2$ are two parameters, the
envelope can be found by the standard approach --- eliminating $t_1$
and $t_2$ from the system $f(x,y,t_1,t_2)=0,
\partial_{t_1}f(x,y,t_1,t_2)=0$, and
$\partial_{t_2}f(x,y,t_1,t_2)=0$.

From the above definition, we can say that the envelope can be
computed by setting $F(x,y,t)=0$ and $\partial_tF(x,y,t)=0$ and then
eliminating $t$; this is the so-called envelope algorithm.

In what follows, we will use the envelope in a polar coordinate
system, and say that $\cF$ is a family of curves given by
$F(r,\theta,t)=0$. The envelope of $\cF$ is the set of points
$(r,\theta)$ so that there is a value of $t$ with both
$F(r,\theta,t)=0$ and $\partial_tF(r,\theta,t)=0$.

\section{Main result}
\label{sect:mainre}

In this section, we present polar descriptions of the boundaries of
the sets $\cB^\circ_n$ and prove that such sets are convex, where
$n\in\set{3,4}$. The result for the case where $n=4$ confirms the
conjecture given by Fernandes \emph{et. al.} in
\cite{Fernandes2024}. We conjecture that the unified,
dimension-independent formulation of the boundaries for sets of
3rd-order and 4th-order Bargmann invariants may extend to the
general case of the $n$th-order Bargmann invariants.
\begin{thrm}\label{th:one}
The equation of the boundary curve of the set $\cB^\circ_n(d)$,
where $n\in\set{3,4}$ and an arbitrary natural number $d\geqslant2$,
representing values $\Delta\in\complex$ that are quantum realizable
by some $n$th-order Bargmann invariant of qudit states is given by
\begin{equation}\label{eq:mainone}
r_n(\theta)e^{\mathrm{i}\theta}=\cos^n\Pa{\tfrac{\pi}{n}}
\sec^n\Pa{\tfrac{\theta-\pi}{n}} e^{\mathrm{i}\theta} =
\Pa{\tfrac{\omega_n+1}{\omega_n+e^{-\mathrm{i}\frac{2\theta}{n}}}}^n,
\end{equation}
where $\omega_n := e^{-\mathrm{i}\frac{2\pi}{n}}$ and $\theta \in
[0, 2\pi)$.
\end{thrm}

We should mention that this result recovers the same findings by
Fernandes \emph{et. al.}, where
$\cB^\circ_3(d)=\cB^\circ_3(2)=\cB^\circ_3=\mathrm{ConvexHull}(\partial\cB^\circ_3)$.
We begin by proving the result for $n=3$. Before the formal proof,
let us recall a well-known fact about cubic equations:  Denote by
$D=-4p^3-27q^2=-108\Pa{\Pa{\frac p3}^3+\Pa{\frac q2}^2}$ the
discriminant of the cubic equation of the form $x^3+px+q=0$, then
\begin{itemize}
\item If $D>0$, the equation $x^3+px+q=0$ has three distinct
\emph{real}
roots;
\item If $D=0$, the equation has a multiple root (either a double root or a triple root).
\item If $D<0$, the equation has one real root and two complex conjugate roots.
\end{itemize}
According to \cite{Fernandes2024}, the rectangular equation of the
boundary curve of $\cB^\circ_3$ is given by $x^2+y^2 =
\Pa{\tfrac{2x+1}3}^3$, which can be equivalently rewritten in polar
equation $3r^{\frac{2}{3}} = 2r \cos\theta + 1$, where $\theta \in
[0, 2\pi)$ and $r>0$ for all $\theta$. Let $s = r^{-\frac{1}{3}}
> 0$, yielding the cubic equation $s^3 - 3s + 2\cos\theta = 0$. Since its discriminant $D=108\sin^2\theta>0$ unless $\theta \in \set{0,
\pi}$, the equation has three distinct real roots. Solving it yields
the positive root $s = 2\cos\Pa{\frac{\theta-\pi}{3}} =
\sec\Pa{\frac{\pi}{3}} \cos\Pa{\frac{\theta-\pi}{3}}$. Therefore,
the corresponding complex point on the boundary is
$$
re^{\mathrm{i}\theta} = s^{-3}e^{\mathrm{i}\theta} =
\cos^3\Pa{\tfrac{\pi}{3}} \sec^3\Pa{\tfrac{\theta-\pi}{3}}
e^{\mathrm{i}\theta}.
$$
In what follows, we will present an independent proof by employing
the numerical range of a complex square matrix and the envelope
method of a family of plane curves.
\begin{proof}[Proof of Theorem~\ref{th:one}]
Note that, based on Lemma~\ref{lem:1}, we can infer that, for any
positive integers $m\geqslant2$, via $\psi_{m+1}\equiv\psi_1$,
\begin{eqnarray}
\cB^\circ_{2m-1}(d) &=&
\bigcup_{\psi_1,\ldots,\psi_m}\Inner{\psi_1}{\psi_2}\prod^m_{j=2}W_d(\out{\psi_{j+1}}{\psi_j}),\label{eq:3}\\
\cB^\circ_{2m}(d) &=&
\bigcup_{\psi_1,\ldots,\psi_m}\prod^m_{j=1}W_d(\out{\psi_{j+1}}{\psi_j}),\label{eq:4}
\end{eqnarray}
where all $\psi_k$'s run over the set of pure states in
$\complex^d$, $\prod^m_{j=1}W_d(\out{\psi_{j+1}}{\psi_j})$ is the
Minkowski product (the Minkowski product of two sets $A$ and $B$ in
$\complex$ is defined by $\set{ab:a\in
A,b\in B}$, denoted by $AB$). The proofs of both identities are provided in Appendix~\ref{app:ndproofs}.\\
(1) If $n=3$, then by Lemma~\ref{lem:1} or Eq.~\eqref{eq:3}, we get
that $\cB^\circ_3(d)
=\bigcup_{\psi_1,\psi_2}\Inner{\psi_1}{\psi_2}W_d(\out{\psi_1}{\psi_2})$,
where
$W_d(\out{\psi_1}{\psi_2})=\Set{z\in\complex:\abs{z}+\abs{z-\Inner{\psi_2}{\psi_1}}\leqslant1}$.
Let $\Inner{\psi_1}{\psi_2}=te^{\mathrm{i}\theta}$, where
$t\in[0,1]$ and $\theta\in[0,2\pi)$. Then
\begin{eqnarray*}
\cB^\circ_3(d)=\bigcup_{t\in[0,1]}\Set{z\in\complex:\abs{z}+\abs{z-t^2}\leqslant
t}=\bigcup_{t\in[0,1]} \cE_t,
\end{eqnarray*}
where $\cE_t:=\Set{z\in\complex:\abs{z}+\abs{z-t^2}\leqslant t}$
whose boundary curve $\partial\cE_t$ can be put in polar form:
$r(\theta)=\frac{t(1-t^2)}{2(1-t\cos\theta)}$. Thus,
$\partial\cB^\circ_3(d)$ can be viewed as the envelope of the family
$\cF$ of ellipses $\partial\cE_t$, defined by
$F(r,\theta,t):=r(1-t\cos\theta)-\frac{t(1-t^2)}2=0$. Then,
$\partial _tF(r,\theta,t)=\frac12(3t^2-2r\cos\theta-1)=0$. By an
envelope algorithm, eliminating $t$ in both $F(r,\theta,t)=0$ and
$\partial_tF(r,\theta,t)=0$, we obtain that the envelope of $\cF$ is
implicitly given by
$$
(8\cos^3\theta)r^3+(12\cos^2\theta-27)r^2+(6\cos\theta)r+1=0.
$$
When we view it as a cubic equation in argument $r$, we solve it to
get three real roots
$\cos^3(\frac\pi3)\sec^3(\frac{\theta-\pi}3),\cos^3(\frac\pi3)\sec^3(\frac{\theta-\pi}3+\frac23\pi),\cos^3(\frac\pi3)\sec^3(\frac{\theta-\pi}3+\frac43\pi)$.
But, only one root $r=\cos^3(\frac\pi3)\sec^3(\frac{\theta-\pi}3)$
is
the desired one.\\
(2) If $n=4$, then still by Eq.~\eqref{eq:4}, we get that
\begin{eqnarray*}
\cB^\circ_4(d) =\bigcup_{\psi_1,\psi_2}
W_d(\out{\psi_2}{\psi_1})W_d(\out{\psi_1}{\psi_2}),
\end{eqnarray*}
where $W_d(\out{\psi_2}{\psi_1})W_d(\out{\psi_1}{\psi_2})$ is the
Minkowski product of two sets in $\complex$. Let
$\Inner{\psi_1}{\psi_2}=te^{\mathrm{i}\phi}$, where $t\in[0,1]$ and
$\phi\in[0,2\pi)$. Then
$W_d(\out{\psi_2}{\psi_1})=\set{z\in\complex:\abs{z}+\abs{z-te^{\mathrm{i}\phi}}\leqslant1}=e^{\mathrm{i}\phi}\set{z\in\complex:\abs{z}+\abs{z-t}\leqslant1}$
and
$W_d(\out{\psi_1}{\psi_2})=e^{-\mathrm{i}\phi}\set{z\in\complex:\abs{z}+\abs{z-t}\leqslant1}$.
Denote by $E_t:=\set{z\in\complex:\abs{z}+\abs{z-t}\leqslant1}$.
Therefore we get that the Minkowski product
$W_d(\out{\psi_2}{\psi_1})W_d(\out{\psi_1}{\psi_2})=E_tE_t\equiv
E^2_t$, which leads to the result
\begin{eqnarray*}
\cB^\circ_4(d) =
\bigcup_{t\in[0,1],\phi\in[0,2\pi)}e^{\mathrm{i}\phi}E_te^{-\mathrm{i}\phi}E_t=\bigcup_{t\in[0,1]}E^2_t.
\end{eqnarray*}
For any fixed $t\in(0,1)$, it is easily seen that $E^2_t
=\bigcup_{z\in\partial E_t}zE_t$, where $z\in\partial E_t$ can be
parameterized as $z=\frac{1-t^2}{2(1-t\cos
\alpha)}e^{\mathrm{i}\alpha}$ for $\alpha\in[0,2\pi)$. Then
\begin{eqnarray*}
E^2_t = \bigcup_{\alpha\in[0,2\pi)} \tfrac{1-t^2}{2(1-t\cos
\alpha)}e^{\mathrm{i}\alpha}E_t,
\end{eqnarray*}
where $\frac{1-t^2}{2(1-t\cos\alpha)}e^{\mathrm{i}\alpha}E_t$ is the
elliptical disk
$$
\Set{z\in\complex: \abs{z}+\abs{z-\tfrac{1-t^2}{2(1-t\cos
\alpha)}te^{\mathrm{i}\alpha}}\leqslant
\tfrac{1-t^2}{2(1-t\cos\alpha)}}
$$
whose boundary can parametrized in polar form
$z=re^{\mathrm{i}\theta}$, where $r$ and $\theta$ can be connected
as
$$
r=\frac{(1-t^2)^2}{4(1-t\cos\alpha)(1-t\cos(\alpha-\theta))}.
$$
Define the family $\tilde\cF$ of curves determined by $\tilde
F(r,\theta,t,\alpha):=(1-t\cos\alpha)(1-t\cos(\alpha-\theta))r-\frac{(1-t^2)^2}4=0$.
Then $\tilde
F_\alpha(r,\theta,t,\alpha)=rt\Br{\sin(\alpha-\theta)+\sin\alpha-t
\sin(2\alpha-\theta)}$. The boundary curve of $E^2_t$ can be
computed as the envelope of the family $\tilde\cF$. In fact, by
envelope algorithm, eliminating $\alpha$ by setting
$F(r,\theta,t,\alpha)=0$ and $F_\alpha(r,\theta,t,\alpha)=0$, we get
that $r=\frac{(1-t^2)^2}{4(1-t\cos\frac\theta2)^2}$ is the polar
equation of $\partial(E^2_t)$, the boundary of the Minkowski product
of $E_t$ with itself.

Now $\cB^\circ_4(d)=\bigcup_{t\in[0,1]}E^2_t$, where
$\partial(E^2_t)$ is parameterized in polar form
$r=\frac{(1-t^2)^2}{4(1-t\cos\frac\theta2)^2}$. Once again, we
define the family $\cG$ of curves by
$G(r,\theta,t):=(1-t\cos\frac\theta2)\sqrt{r} - \frac{1-t^2}2=0$.
Then $\partial_tG(r,\theta,t)=t-\sqrt{r}\cos\frac\theta2$. Now the
boundary $\partial\cB^\circ_4(d)$, as the envelope of the family
$\cG$ of curves, can be computed as by setting $G(r,\theta,t)=0$ and
$\partial_tG(r,\theta,t)=0$ by envelope algorithm. Eliminating $t$,
we get that $r=\frac1{(\sin\frac\theta4\pm\cos\frac\theta4)^4}$. But
only
$r=\frac1{(\sin\frac\theta4+\cos\frac\theta4)^4}=\cos^4(\frac\pi4)\sec^4(\frac{\theta-\pi}4)$
is indeed the envelope.
\end{proof}

We remark here that
\begin{eqnarray}\label{eq:2}
\cB^\circ_3(d) = \bigcup_{t\in[0,1]}tE_t\text{ and }\cB^\circ_4(d) =
\bigcup_{t\in[0,1]}E^2_t,
\end{eqnarray}
where $E_t=\set{z\in\complex:\abs{z}+\abs{z-t}\leqslant1}$ is the
elliptical disk. Clearly from the right hand sides in
Eq.~\eqref{eq:2}, they are independent of the dimension
$d\geqslant2$. The result obtained for $\cB^\circ_4(d)$ in
Theorem~\ref{th:one} \emph{confirms} the conjecture given by
Fernandes \emph{et al.} \cite{Fernandes2024}, where the authors
conjectured that
$r(\theta)=\frac1{(\sin\frac\theta4+\cos\frac\theta4)^4}$, the
boundary of the set of circulant matrices $\cB_4|_{\text{circ}}$ (in
their notation), is equal to the boundary $\cB_4$ of the set of
\emph{all} 4th-order Bargmann invariants. In fact, they conjectured
that, for $n\geqslant3$, in their notations,
$$
\cB_n=\cB_{n,2}=\mathrm{ConvexHull}(\cB_n)=\mathrm{ConvexHull}(\cB_n|_{\mathrm{circ}}).
$$

Note that the inner product is a dimension-independent quantity when
restricted to the set of pure states, as the inner product between
any two pure states $\ket{\psi}$ and $\ket{\phi}$ in $\complex^d$
always lies within the closed unit disk in $\complex$ (i.e.,
$\abs{\Inner{\psi}{\phi}}\leqslant1$), regardless of the dimension
$d\geqslant2$. This property holds universally across all finite
dimensions. Based on the results for the 3rd and 4th-order Bargmann
invariants, it is reasonable to believe that the sets of all
$n$th-order Bargmann invariants in the underlying space $\complex^d$
are dimension independent. This is consistent with Fernandes
\emph{et al.}'s conjecture \cite{Fernandes2024}. We then rephrase it
as follows:
\begin{itemize}
\item $\cB^\circ_n(d)$ is independent of $d\geqslant2$ for arbitrary positive integer $n\geqslant3$. In other words,
$\cB^\circ_n(d)=\cB^\circ_n(2)=\cB^\circ_n$ for all $d\geqslant2$
and $n\geqslant3$.
\end{itemize}
Let $\cR_n$ denote the closed region enclosed by the curve
$r_n(\theta)e^{\mathrm{i}\theta}=\cos^n(\frac\pi
n)\sec^n(\frac{\theta-\pi}n)e^{\mathrm{i}\theta}$ in $\complex$. It
is easily seen that $\cR_n\subseteq\cB^\circ_n$ by
Proposition~\ref{prop:2}. Building upon Fernandes \emph{et al.}'s
conjecture and Theorem~\ref{th:one}, we are led to formulate the
subsequent
conjecture:\\~\\
\textbf{Conjecture 1.} {\em For $n\geqslant3$, it holds that
$\cB^\circ_n=\cR_n$. That is, $\partial\cB^\circ_n$ can be described
analytically in polar form of the boundary of $\cR_n$:
$r_n(\theta)=\cos^n(\frac\pi n)\sec^n(\frac{\theta-\pi}n)$, where
$\theta\in[0,2\pi)$.}

Both Fernandes \emph{et al.}'s conjecture and Conjecture 1 were
confirmed for $n=3$ \cite{Fernandes2024}. Additionally, they
provided numerical support for the case $n= 4$, which is ultimately
proven by our Theorem~\ref{th:one}.

As a direct implication by assuming that the above conjectures are
correct, i.e., $\cB^\circ_n=\cR_n$, we have
$$
\lim_{n \to \infty} \cos^n\Pa{\tfrac{\pi}{n}}
\sec^n\Pa{\tfrac{\theta-\pi}{n}} e^{\mathrm{i}\theta} =
e^{\mathrm{i}\theta} \quad \text{for all } \theta \in [0, 2\pi],
$$
indicating that the family of boundary curves $\partial \cB^\circ_n$
converges to the unit circle as $n \to \infty$:
$$
\cB^\circ_3 \subsetneq \cB^\circ_4 \subsetneq \cdots \subsetneq
\cB^\circ_n \subsetneq \cdots \subsetneq \set{\abs{z}=1 : z \in
\complex}.
$$
This containment relationship suggests that higher-order Bargmann
invariants demonstrate enhanced effectiveness as witnesses of
quantum imaginarity. The boundary curves $\partial\cR^\circ_n$ are
plotted in Fig.~\ref{fig:bcurves}.
\begin{figure}[ht]\centering
{\begin{minipage}[b]{1\linewidth}
\includegraphics[width=0.7\textwidth]{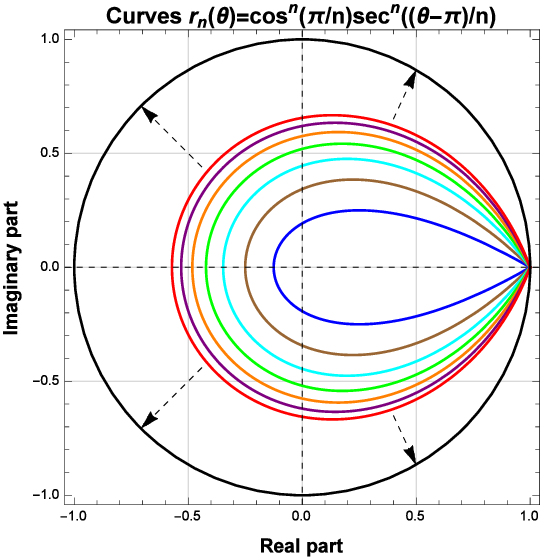}
\end{minipage}}
\caption{(Color Online) The family of curves in polar form
$r_n(\theta)=\cos^n(\frac\pi n)\sec^n(\frac{\theta-\pi}n)$, where
$n\in\set{3(\text{blue}),4(\text{brown}),5(\text{cyan}),6(\text{green}),7(\text{orange}),8(\text{purple}),9(\text{red})}$.
Note that $n\in\set{3,4}$, $r_n(\theta)=\cos^n(\frac\pi
n)\sec^n(\frac{\theta-\pi}n)$ is indeed shown to be
$\partial\cB^\circ_n$ by Theorem~\ref{th:one}. The black curve is
the unit circle. Here the horizontal axis means the real part
$\mathrm{Re}(\Tr{\psi_1\cdots\psi_n})$ and the vertical axis stands
for the imaginary part
$\mathrm{Im}(\Tr{\psi_1\cdots\psi_n})$.}\label{fig:bcurves}
\end{figure}

\begin{cor}
It holds that
\begin{eqnarray}\label{eq:largimpart}
\cR_n\subset[-\cos^n\Pa{\tfrac\pi n},1] \times [-\tau_n,\tau_n],
\end{eqnarray}
where $\tau_n=\cos^n\Pa{\tfrac\pi
n}\sec^{n-1}\Pa{\tfrac\pi{2(n-1)}}$.
\end{cor}

\begin{proof}
Denote by $f_n(\theta):=\sec^n\Pa{\tfrac{\theta-\pi}n}\sin \theta$
and let $\tau_n:=\cos^n\Pa{\tfrac\pi n}\max_{\theta\in[0,\pi]}
f_n(\theta)$. From Eq.~\eqref{eq:mainone}, it follows that
$\cR_n\subset[-\cos^n\Pa{\tfrac\pi n},1] \times [-\tau_n,\tau_n]$.
Next we identify $\tau_n$. In fact,
$f'_n(\theta)=\sec^{n+1}\Pa{\tfrac{\theta-\pi}n}\cos\Pa{\tfrac\pi
n+\Pa{1-\tfrac1n}\theta}=0$ implies that $\cos\Pa{\frac\pi
n+\Pa{1-\frac1n}\theta}=0$ because
$\sec\Pa{\frac{\theta-\pi}n}=\sec\Pa{\frac{\abs{\theta-\pi}}n}>0$
(here $\frac{\abs{\theta-\pi}}n\leqslant\frac\pi n$). Thus $\frac\pi
n+\Pa{1-\frac1n}\theta=\frac\pi2$, then $\theta=
\frac{n-2}{n-1}\frac\pi2<\frac\pi2$ is the stationary point of
$f_n(\theta)$ in $[0,\pi]$. Based on the above reasoning, we get
that
\begin{eqnarray*}
\tau_n &=& \cos^n\Pa{\tfrac\pi n}f_n(\tfrac{n-2}{n-1}\tfrac\pi2)\\
&=& \cos^n\Pa{\tfrac\pi n}\sec^{n-1}\Pa{\tfrac\pi{2(n-1)}},
\end{eqnarray*}
as desired.
\end{proof}
We remark here that $\tau_n$ could potentially predicts the largest
imaginary part of $n$th-order Bargmann invariants if it were true
that $\cR_n=\cB^\circ_n$. In fact,
$\tau_3=\tfrac14,\tau_4=\tfrac2{3\sqrt{3}},\tau_5=\tfrac{(3-2\sqrt{2})(11+5\sqrt{5})}8,\tau_6=\tfrac{27}{100}\sqrt{5(50-22\sqrt{5})}$.
For illustration, considering specific values of $n\in\set{3,4}$, we
present another approach to $\tau_3$ and $\tau_4$ in
Appendix~\ref{app:B}. Furthermore, it is observed that all
$n$th-order Bargmann invariants along the curve $r_n(\theta)$
exhibiting the largest possible imaginary part invariably possess a
non-zero real component.

\begin{prop}\label{prop:2}
The curve $r_n(\theta)=\cos^n(\frac\pi n)\sec^n(\frac{\theta-\pi}n)$
determined by Eq.~\eqref{eq:mainone}, can be attained by a family of
single-parameter qubit pure states.
\end{prop}

\begin{proof}
Denote $\omega_n=e^{-\mathrm{i}\frac{2\pi}{n}}$. Consider the
$n$-tuple of qubit pure states $\Psi=(\psi_0,\ldots,\psi_{n-1})$
where $\psi_k\equiv\out{\psi_k}{\psi_k}$ and $\ket{\psi_k(t)}:=\sin
t\ket{0}+\bar{\omega}_n^k\cos t\ket{1}$, where $k=0,1,\cdots,n-1$.
Then
$\Inner{\psi_k}{\psi_{k+1}}=\sin^2t+\bar{\omega}_n\cos^2t=\Inner{\psi_{n-1}}{\psi_0}$
for all $k=0,1,\ldots,n-2$. Therefore,
\begin{eqnarray*}
\Delta_n&=&\Tr{\psi_0\cdots\psi_{n-1}}=\Pa{\sin^2t+\bar\omega_n\cos^2t}^n\\
    &=&\Br{\Pa{\sin^2t+\cos\Pa{\tfrac{2\pi}n}\cos^2t}+\mathrm{i}\Pa{\sin\Pa{\tfrac{2\pi}n}\cos^2t}}^n,
\end{eqnarray*}
This can be expressed as $\Delta_n=\cos^n\Pa{\tfrac\pi n}
\sec^n\Pa{\tfrac{\theta-\pi}n}e^{\mathrm{i}\theta}$ for some angle
$\theta$, determined by
$$
\tan\Pa{\tfrac\theta
n}=\frac{\sin\Pa{\tfrac{2\pi}{n}}}{\cos\Pa{\tfrac{2\pi}{n}}+\tan^2t}.
$$
Indeed, for this tuple $\Psi=\Pa{\psi_0,\ldots,\psi_{n-1}}$, the
Bargmann invariant $\Delta_n(\Psi)$ is of the polar form
$\Delta_n=\abs{\Delta_n}e^{\mathrm{i}\theta}$, where
\begin{eqnarray*}
\abs{\Delta_n}=\Br{\Pa{\sin^2t+\cos\Pa{\tfrac{2\pi}n}\cos^2t}^2+\Pa{\sin\Pa{\tfrac{2\pi}n}\cos^2t}^2}^{\frac
n2},
\end{eqnarray*}
and its phase $\theta$ satisfying $\frac{\sin\Pa{\frac\theta
n}}{\cos\Pa{\frac\theta
n}}=\frac{\sin\Pa{\frac{2\pi}{n}}\cos^2t}{\sin^2t+\cos\Pa{\frac{2\pi}{n}}\cos^2t}$,
i.e.,
\begin{eqnarray*}
\frac{\sin\Pa{\frac\theta n}}{\cos\Pa{\frac\theta n}}=
\frac{2\sin\Pa{\tfrac\pi n}\cos\Pa{\tfrac\pi
n}\cos^2t}{1-2\sin^2\Pa{\tfrac\pi n}\cos^2t}.
\end{eqnarray*}
This leads to
\begin{eqnarray}\label{eq:D3}
\cos^2t =\frac{\sin\Pa{\tfrac\theta n}}{2\sin\Pa{\tfrac\pi
n}\cos\Pa{\tfrac{\theta-\pi}n}},
\end{eqnarray}
where $\sin\Pa{\tfrac\pi n}>0$ and $\cos\Pa{\tfrac{\theta-\pi}n}>0$,
because $0<\tfrac\pi n\leqslant\tfrac\pi3$ and
$\tfrac{\abs{\theta-\pi}}n\leqslant\tfrac\pi n$. By substituting
Eq.~\eqref{eq:D3} into the expression of $\abs{\Delta_n}$, we have
\begin{eqnarray*}
\abs{\Delta_n}&=&\Br{\Pa{\sin^2t+\cos\Pa{\tfrac{2\pi}n}\cos^2t}^2+\Pa{\sin\Pa{\tfrac{2\pi}n}\cos^2t}^2}^{\frac
n2}\\
&=&\Br{\Pa{1-2\sin^2\Pa{\tfrac\pi n}\cos^2t}^2+\Pa{\sin\Pa{\tfrac{2\pi}n}\cos^2t}^2}^{\frac n2}\\
&=&\Br{\frac{\Pa{\sin\Pa{\tfrac{2\pi} n}\cos\Pa{\tfrac\theta
n}}^2+\Pa{\sin\Pa{\tfrac{2\pi} n}\sin\Pa{\tfrac\theta
n}}^2}{\Pa{2\sin\Pa{\tfrac\pi
n}\cos\Pa{\tfrac{\theta-\pi}n}}^2}}^\frac{n}{2}\\
&=&\Pa{\frac{\sin\Pa{\tfrac{2\pi}n}}{2\sin\Pa{\tfrac\pi
n}\cos\Pa{\tfrac{\theta-\pi}n}}}^n =\cos^n\Pa{\tfrac\pi
n}\sec^n\Pa{\tfrac{\theta-\pi}n}.
\end{eqnarray*}
Based on the above reasoning, we conclude that the curves described
by Eq.~\eqref{eq:mainone} are quantum realizable.
\end{proof}
In the proof of Proposition~\ref{prop:2}, the constructed pure
states $\ket{\psi_k(t)}=\sin t\ket{0}+\bar{\omega}_n^k\cos
t\ket{1}$, where $k=0,1,\cdots,n-1$, could be called
\emph{Oszmaniec-Brod-Galv\~{a}o's states} as this family of states
is a variation of the family of states that has already appeared in
\cite{Oszmaniec2024}.

\begin{thrm}
The region $\cR_n$ enclosed by the curve
$r_n(\theta)=\cos^n(\frac\pi n)\sec^n(\frac{\theta-\pi}n)$ is
convex.
\end{thrm}
The idea of proof is originated from the convexity argument provided
by Fernandes \emph{et. al.} \cite{Fernandes2024}.

\begin{proof}
Note that the graph of this curve is symmetric with respect to the
real axis. Then we show that the graph of this curve
$r_n(\theta)=\cos^n\Pa{\frac\pi n}\sec^n\Pa{\frac{\theta-\pi}n}$ is
\emph{concave} on the open interval $(0,\pi)$. This implies that the
region below this curve must be a convex region. By the symmetry of
this curve with respect to the real axis, $r_n(\theta)$ is
\emph{convex} on the open interval $(\pi,2\pi)$. Both together form
the region $\cR_n$. Let us rewrite it as the parametric equation in
parameter $\theta$:
\begin{eqnarray*}
\begin{cases}
x(\theta) =\cos^n\Pa{\frac\pi n}\sec^n\Pa{\frac{\theta-\pi}n}\cos\theta,\\
y(\theta) =\cos^n\Pa{\frac\pi
n}\sec^n\Pa{\frac{\theta-\pi}n}\sin\theta.
\end{cases}
\end{eqnarray*}
In fact, for $n\geqslant4$ and $\theta\in(0,\pi)$,
\begin{eqnarray*}
\frac{\dif^2 y}{\dif x^2} = -\tfrac{n-1}n\sec^n\Pa{\tfrac\pi
n}\cos^{n+1}\Pa{\tfrac{\theta-\pi}n}\csc^3\Pa{\tfrac{\pi+(n-1)\theta}n},
\end{eqnarray*}
which is negative because all factors on the right hand side are
positive up to the minus sign. In summary, $\frac{\dif^2 y}{\dif
x^2}<0$ on $(0,\pi)$, implying that the curve
$r_n(\theta)=\cos^n\Pa{\frac\pi n}\sec^n\Pa{\frac{\theta-\pi}n}$ is
concave on $[0,\pi]$. Furthermore, it is convex on $[\pi,2\pi]$ by
the symmetry with respect to the real axis. Therefore the region
$\cR_n$ enclosed by the curve $r_n(\theta)=\cos^n\Pa{\frac\pi
n}\sec^n\Pa{\frac{\theta-\pi}n}$, where $\theta\in[0,2\pi]$, is
convex.
\end{proof}

\section{Discussion and outlook}
\label{sect:evidence}

For Fernandes \emph{et al.}'s conjecture, Fernandes \emph{et al.}
\cite{Fernandes2024} have confirmed that $\cB^\circ_3(d) =
\cB^\circ_3(2)=\cB^\circ_3$ and our Theorem~\ref{th:one} have
already confirmed that $\cB^\circ_4(d) =
\cB^\circ_4(2)=\cB^\circ_4$. The core of their conjecture
essentially asserts that, for any $n$-tuple $(\psi_1,\ldots,\psi_n)$
of qudit pure states in $\complex^d(d\geqslant2)$, there must be an
$n$-tuple $(\tilde\psi_1,\ldots,\tilde\psi_n)$ of qubit pure states
in $\complex^2$ such that
$\Tr{\psi_1\cdots\psi_n}=\Tr{\tilde\psi_1\cdots\tilde\psi_n}$, where
each $\psi_k\equiv\out{\psi_k}{\psi_k}$ and
$\tilde\psi_k\equiv\out{\tilde\psi_k}{\tilde\psi_k}$. Clearly,
constructing such correspondence is highly non-trivial. Building on
this observation, from Eq.~\eqref{eq:1}, we observe that
$\cB^\circ_n(d)$ can be expressed as a union of a family of
elliptical disks (denoted by $\cE_{\psi_1,\ldots,\psi_{n-1}}$). If
we can construct a family of qubit pure states
$(\tilde\psi_1,\ldots,\tilde\psi_{n-1})$ such that
$\cE_{\psi_1,\ldots,\psi_{n-1}}=\cE_{\tilde\psi_1,\ldots,\tilde\psi_{n-1}}$,
then we can demonstrate that
$\cB^\circ_n(d)=\cB^\circ_n(2)=\cB^\circ_n$.

Indeed, from this perspective, mathematical induction could
potentially come into play because the assumption
$\cB^\circ_{n-1}(d)=\cB^\circ_{n-1}(2)=\cB^\circ_{n-1}$ will
guarantee that for any $(n-1)$-tuple $(\psi_1,\ldots,\psi_{n-1})$ of
pure states in $\complex^d$, there must be a $(n-1)$-tuple
$(\tilde\psi_1,\ldots,\tilde\psi_{n-1})$ of qubit pure states in
$\complex^2$ such that
$\Tr{\psi_1\cdots\psi_{n-1}}=\Tr{\tilde\psi_1\cdots\tilde\psi_{n-1}}$.
If, moreover,
$\abs{\prod^{n-2}_{j=1}\Inner{\psi_j}{\psi_{j+1}}}=\abs{\prod^{n-2}_{j=1}\Inner{\tilde\psi_j}{\tilde\psi_{j+1}}}$,
then, together with
$\Tr{\psi_1\cdots\psi_{n-1}}=\Tr{\tilde\psi_1\cdots\tilde\psi_{n-1}}$,
we have
$\cE_{\psi_1,\ldots,\psi_{n-1}}=\cE_{\tilde\psi_1,\ldots,\tilde\psi_{n-1}}$.
This implies that
\begin{eqnarray*}
\cB^\circ_n(d)&=&\cup_{\psi_1,\ldots,\psi_{n-1}}\cE_{\psi_1,\ldots,\psi_{n-1}}\\
&\subseteq&\cup_{\tilde\psi_1,\ldots,\tilde\psi_{n-1}}\cE_{\tilde\psi_1,\ldots,\tilde\psi_{n-1}}=\cB^\circ_n(2).
\end{eqnarray*}
Therefore $\cB^\circ_n(d)=\cB^\circ_n(2)$ follows immediately from
$\cB^\circ_n(2)\subseteq\cB^\circ_n(d)$, a trivial fact. We leave
this as an open problem for future research.

In this study, we have made significant progress in addressing the
theoretical challenge of determining the boundaries of sets of
Bargmann invariants, which are crucial unitary invariants in quantum
information theory. By developing a unified, dimension-independent
formulation, we have successfully characterized the sets of 3rd and
4th Bargmann invariants, thereby confirming a recent conjecture
proposed by Fernandes \emph{et al.} for the 4th-order case. Our
results not only provide a deeper understanding of the geometric and
algebraic structures underlying these invariants but also suggest a
potential extension of our formulation to the general case of
$n$th-order Bargmann invariants.

These findings have some implications for quantum mechanics, as they
elucidate fundamental physical limits and open avenues for exploring
the role of Bargmann invariants in quantum information processing.
In future research, we will further investigate the applications of
Bargmann invariants, particularly in detecting entanglement, as
referenced in \cite{Lin2024}. Additionally, future work could focus
on rigorously proving the conjectured generalization to higher-order
invariants and expanding their applications in quantum computing,
quantum communication, and other related fields. This study thus
marks a significant step toward unifying the theoretical framework
of Bargmann invariants and harnessing their potential for advancing
quantum technologies.

\textit{Acknowledgments.}--- We are deeply grateful for the
significant and crucial comments provided by the anonymous
reviewers, as they have been instrumental in enhancing the quality
of our paper. We also would like to thank Dr. Rafael Wagner for his
feedback. This research is supported by Zhejiang Provincial Natural
Science Foundation of China under Grants No. LZ23A010005, and No.
LZ24A050005, and by NSFC under Grants No.11971140, and No. 12175147.

\emph{Note added:} Recently, we became aware of a related work by Li
and Tan \cite{Li2024}. Their findings complement our results and
provide additional insights about Bargmann invariants.



\begin{widetext}
\newpage
\appendix

\section{The proofs of both identities \eqref{eq:3} and \eqref{eq:4} in the main text}
\label{app:ndproofs}

For any positive integer $n\geqslant3$, it can be written as
$n=2m-1$ or $2m$ for a positive integer $m\geqslant2$. Thus
$\cB^\circ_n(d)=\cB^\circ_{2m-1}(d)$ or $\cB^\circ_{2m}(d)$ for
$m\geqslant2$.

\begin{itemize}
\item {\bf Proof of the identity \eqref{eq:3}}. From
Lemma~\ref{lem:1}, it follows that
\begin{eqnarray*}
\cB^\circ_{2m-1}(d)
&=&\bigcup_{\varphi_1,\ldots,\varphi_{2(m-1)}}\Inner{\varphi_1}{\varphi_2}\Inner{\varphi_2}{\varphi_3}\Inner{\varphi_3}{\varphi_4}\cdots\Inner{\varphi_{2(m-2)}}{\varphi_{2m-3}}\Inner{\varphi_{2m-3}}{\varphi_{2(m-1)}}W_d(\out{\varphi_1}{\varphi_{2(m-1)}})\\
&=&
\bigcup_{\varphi_1,\varphi_{2j}:j=1,2,\ldots,m-1}\Inner{\varphi_1}{\varphi_2}\Pa{\prod^m_{j=3}\bigcup_{\varphi_{2j-3}}\Innerm{\varphi_{2j-3}}{\out{\varphi_{2(j-1)}}{\varphi_{2(j-2)}}}{\varphi_{2j-3}}}W_d(\out{\varphi_1}{\varphi_{2(m-1)}})\\
&=&
\bigcup_{\varphi_1,\varphi_{2j}:j=1,2,\ldots,m-1}\Inner{\varphi_1}{\varphi_2}\Pa{\prod^m_{j=3}W_d(\out{\varphi_{2(j-1)}}{\varphi_{2(j-2)}})}W_d(\out{\varphi_1}{\varphi_{2(m-1)}})
\end{eqnarray*}
By setting
$(\varphi_1,\varphi_2,\varphi_4,\ldots,\varphi_{2(m-1)})=(\psi_1,\psi_2,\ldots,\psi_m)$
and $\psi_{m+1}\equiv\psi_1$, we get that
\begin{eqnarray*}
\cB^\circ_{2m-1}(d)=
\bigcup_{\psi_1,\ldots,\psi_m}\Inner{\psi_1}{\psi_2}\prod^m_{j=2}W_d(\out{\psi_{j+1}}{\psi_j}).
\end{eqnarray*}

\item {\bf Proof of the identity \eqref{eq:4}}. From
Lemma~\ref{lem:1}, it follows that
\begin{eqnarray*}
\cB^\circ_{2m}(d)&=&\bigcup_{\varphi_1,\ldots,\varphi_{2m-1}}
\Inner{\varphi_1}{\varphi_2}\Inner{\varphi_2}{\varphi_3}\Inner{\varphi_3}{\varphi_4}\cdots\Inner{\varphi_{2(m-1)}}{\varphi_{2m-1}}W_d(\out{\varphi_1}{\varphi_{2m-1}})\\
&=&
\bigcup_{\varphi_j:j=1,3,\ldots,2m-1}\Pa{\prod^m_{j=2}\bigcup_{\varphi_{2(j-1)}}\Innerm{\varphi_{2(j-1)}}{\out{\varphi_{2j-1}}{\varphi_{2j-3}}}{\varphi_{2(j-1)}}}W_d(\out{\varphi_1}{\varphi_{2m-1}})\\
&=&
\bigcup_{\varphi_j:j=1,3,\ldots,2m-1}\Pa{\prod^m_{j=2}W_d(\out{\varphi_{2j-1}}{\varphi_{2j-3}})}W_d(\out{\varphi_1}{\varphi_{2m-1}})
\end{eqnarray*}
By setting
$(\varphi_1,\varphi_3,\varphi_5,\ldots,\varphi_{2m-1})=(\psi_1,\psi_2,\ldots,\psi_m)$
and $\psi_{m+1}\equiv\psi_1$, we get that
\begin{eqnarray*}
\cB^\circ_{2m}(d)=
\bigcup_{\psi_1,\ldots,\psi_m}\prod^m_{j=1}W_d(\out{\psi_{j+1}}{\psi_j}).
\end{eqnarray*}
\end{itemize}

\section{The actual largest imaginary parts $\tau_3$ and $\tau_4$}
\label{app:B}

In what follows, let us focus on the qubit system for convenience of
analytical calculations. For any qubit state
$\rho\in\density{\complex^2}$, the set of density matrices acting on
$\complex^2$, we can decompose it as
$\rho=\rho(\bsr)=\frac12(\I_2+\bsr\cdot\boldsymbol{\sigma})$ using
the Bloch representation. Define
$\bsP_n=\rho_1\rho_2\cdots\rho_n=2^{-n}\Pa{p^{(n)}_0\I+\bsp^{(n)}\cdot\boldsymbol{\sigma}}$,
where $p^{(1)}_0=1$ and $\bsp^{(1)}=\bsr_1$. Then, by recursion,
$\bsP_{n+1} = \bsP_n\rho_{n+1}$, implying the following update
rules:
\begin{eqnarray*}
\begin{cases}
p^{(n+1)}_0 &= p^{(n)}_0+\Inner{\bsp^{(n)}}{\bsr_{n+1}},\\
\bsp^{(n+1)} &=
p^{(n)}_0\bsr_{n+1}+\bsp^{(n)}+\mathrm{i}\bsp^{(n)}\times\bsr_{n+1}.
\end{cases}
\end{eqnarray*}
The Bargmann invariant for the tuple $(\rho_1,\ldots,\rho_n)$, where
$\rho_k=\rho(\bsr_k)$ with $\abs{\bsr_k}=1$ for each $k$, is thus
given by $\Delta_n:=\Tr{\bsP_n} = 2^{1-n}p^{(n)}_0=x+\mathrm{i}y$,
where $x,y\in \real$. We now provide explicit expressions for
specific cases $n=3,4,5$.
\begin{enumerate}
\item[(1)] For $n=3$, both the real and imaginary parts of $\Delta_n(\Psi)=\Tr{\rho_1\cdots\rho_n}$ are given by
\begin{eqnarray}\label{eq:A1}
\begin{cases}
x = \frac14\Pa{1+\sum_{1\leqslant i<j\leqslant3}\Inner{\bsr_i}{\bsr_j}},\\
y =
\frac14\Inner{\bsr_1\times\bsr_2}{\bsr_3}=\frac14\det(\bsr_1,\bsr_2,\bsr_3).
\end{cases}
\end{eqnarray}

\item[(2)] For $n=4$, both the real and imaginary parts of $\Delta_n(\Psi)=\Tr{\rho_1\cdots\rho_n}$ are given by
\begin{eqnarray}\label{eq:A2}
\begin{cases}
x = \frac18\Pa{(1+\Inner{\bsr_1}{\bsr_2})(1+\Inner{\bsr_3}{\bsr_4}) - (1-\Inner{\bsr_1}{\bsr_3})(1-\Inner{\bsr_2}{\bsr_4}) + (1+\Inner{\bsr_1}{\bsr_4})(1+\Inner{\bsr_2}{\bsr_3})},\\
y = \frac18\det(\bsr_1+\bsr_2,\bsr_2+\bsr_3,\bsr_3+\bsr_4).
\end{cases}
\end{eqnarray}

\item[(3)] For $n=5$, both the real and imaginary parts of $\Delta_n(\Psi)=\Tr{\rho_1\cdots\rho_n}$ are given by
\begin{eqnarray}\label{eq:A3}
\begin{cases}
x = & \frac1{16}\Big(1+\sum_{1\leqslant i<j\leqslant 5}\Inner{\bsr_i}{\bsr_j}+\Inner{\bsr_1}{\bsr_2}\Inner{\bsr_3}{\bsr_4}-\Inner{\bsr_1}{\bsr_3}\Inner{\bsr_2}{\bsr_4}\\
&+\Inner{\bsr_1}{\bsr_4}\Inner{\bsr_2}{\bsr_3} + (\Inner{\bsr_2}{\bsr_3}+\Inner{\bsr_2}{\bsr_4}+\Inner{\bsr_3}{\bsr_4})\Inner{\bsr_1}{\bsr_5}\\
&+(-\Inner{\bsr_1}{\bsr_3}-\Inner{\bsr_1}{\bsr_4}+\Inner{\bsr_3}{\bsr_4})\Inner{\bsr_2}{\bsr_5}\\
&+(\Inner{\bsr_1}{\bsr_2}-\Inner{\bsr_1}{\bsr_4}-\Inner{\bsr_2}{\bsr_4})\Inner{\bsr_3}{\bsr_5}\\
&+(\Inner{\bsr_1}{\bsr_2}+\Inner{\bsr_1}{\bsr_3}+\Inner{\bsr_2}{\bsr_3})\Inner{\bsr_4}{\bsr_5}\Big),\\
y = & \frac1{16}\Big(\sum_{1\leqslant i<j<k\leqslant5}\Inner{\bsr_i\times\bsr_j}{\bsr_k}+\Inner{\bsr_2}{\bsr_3}\Inner{\bsr_1\times\bsr_4}{\bsr_5}\\
&-\Inner{\bsr_1}{\bsr_3}\Inner{\bsr_2\times\bsr_4}{\bsr_5} +\Inner{\bsr_1}{\bsr_2}\Inner{\bsr_3\times\bsr_4}{\bsr_5}\\
&+\Inner{\bsr_4}{\bsr_5}\Inner{\bsr_1\times\bsr_2}{\bsr_3}\Big).
\end{cases}
\end{eqnarray}
\end{enumerate}

Indeed, for $\Delta_n=x+\mathrm{i}y$, the largest imaginary part of
$\Delta_n$ is $\max\abs{y}$. Next we check that
$\tau_n=\max\abs{y}$. Although the fact that $\tau_3=\frac14$ is the
actual largest imaginary part of all elements in $\cB^\circ_3$ is
obtained already in \cite{Fernandes2024}, here we reprove it for
completeness using a different method.
\begin{prop}\label{prop:1}
For any three vectors $\set{\bsa,\bsb,\bsc}$ in $\real^3$, it holds
that
\begin{equation}\label{eq:C1}
-\abs{\bsa}\abs{\bsb}\abs{\bsc}\leqslant\det(\bsa,\bsb,\bsc)\leqslant
\abs{\bsa}\abs{\bsb}\abs{\bsc}.
\end{equation}
The above inequalities are saturated if and only if the pairwise
orthogonality of $\set{\bsa,\bsb,\bsc}$ is true, i.e.,
$\Inner{\bsa}{\bsb}=\Inner{\bsa}{\bsc}=\Inner{\bsb}{\bsc}=0$.
Moreover, $\det(\bsa,\bsb,\bsc)= \abs{\bsa}\abs{\bsb}\abs{\bsc}$ if
and only if
$\Inner{\bsa}{\bsb}=\Inner{\bsa}{\bsc}=\Inner{\bsb}{\bsc}=0$ and
$\set{\bsa,\bsb,\bsc}$ forms the right-handed system.
\end{prop}

\begin{proof}
Clearly, $\det(\bsa,\bsb,\bsc)=0$ if and only if
$\set{\bsa,\bsb,\bsc}$ is linearly dependent. Thus we need only
consider the case where $\set{\bsa,\bsb,\bsc}$ is linearly
independent. In this case, there is a triangle pyramid with three
edges $\set{\bsa,\bsb,\bsc}$ and a common vertex at the origin. Note
that
$$
\det(\bsa,\bsb,\bsc)=\abs{\bsa}\abs{\bsc}\abs{\bsc}\det(\bsa',\bsb',\bsc').
$$
It suffices to show that $\abs{\det(\bsa',\bsb',\bsc')}\leqslant1$
for three unit vectors $\set{\bsa',\bsb',\bsc'}$. Thus without loss
of generality, we can assume that
$\abs{\bsa}=\abs{\bsb}=\abs{\bsc}=1$.  Recall that in $\real^3$, a
parallelopiped with three edges $\set{\bsa,\bsb,\bsc}$, its volume
is given by
$\abs{\bsa\cdot(\bsb\times\bsc)}=\abs{\det(\bsa,\bsb,\bsc)}$. The
volume $V(\bsa,\bsb,\bsc)$ of a triangle pyramid is given by
$\abs{\det(\bsa,\bsb,\bsc)}=6V(\bsa,\bsb,\bsc)$. Assume that the
height of this pyramid is $h$, then
$$
V(\bsa,\bsb,\bsc) =\frac13 S_{\max}h =\frac{\sqrt{3}}4(1-h^2)h
$$
where $h\leqslant1$. Thus, the maximal volume is given by
$V_{\max}=\frac16$, which is attained at $h=\frac1{\sqrt{3}}$.
Therefore, the maximal determinant
$\abs{\det(\bsa,\bsb,\bsc)}=6\frac16=1$. For general
$\set{\bsa,\bsb,\bsc}$, it holds that
$$
\abs{\det(\bsa,\bsb,\bsc)}=\abs{\bsa}\abs{\bsb}\abs{\bsc}.
$$
when $\Inner{\bsa}{\bsb}=\Inner{\bsa}{\bsc}=\Inner{\bsb}{\bsc}=0$.
\end{proof}

\begin{prop}
It holds that
\begin{enumerate}
\item[(1)] The real part of 3rd order Bargmann invariant
$x$ ranges over $[-\frac18,1]$.
\item[(2)]  The imaginary part of 3rd order Bargmann invariant
$y$ ranges over $[-\frac14,\frac14]$.
\end{enumerate}
\end{prop}

\begin{proof}
Note that
$\Inner{\bsr_1}{\bsr_2}+\Inner{\bsr_2}{\bsr_3}+\Inner{\bsr_1}{\bsr_3}=4x-1$
and $\det(\bsr_1,\bsr_2,\bsr_3)=4y$ by Eq.~\eqref{eq:A1}. We also
let $\bsR=(\bsr_1,\bsr_2,\bsr_3)$. Then
\begin{eqnarray*}
\bsR^\t\bsR= \Pa{\begin{array}{ccc}
                   1 & \Inner{\bsr_1}{\bsr_2} & \Inner{\bsr_1}{\bsr_3} \\
                   \Inner{\bsr_1}{\bsr_2} & 1 & \Inner{\bsr_2}{\bsr_3} \\
                   \Inner{\bsr_1}{\bsr_3} & \Inner{\bsr_2}{\bsr_3} & 1
                 \end{array}
}.
\end{eqnarray*}
Note that
$\Inner{\bsr_1}{\bsr_2}=\cos\gamma,\Inner{\bsr_1}{\bsr_3}=\cos\beta$,
and $\Inner{\bsr_2}{\bsr_3}=\cos\alpha$, where
$\alpha,\beta,\gamma\in[0,\pi]$. Based on this, we get that
\begin{eqnarray*}
4x -1&=& \cos\alpha+\cos\beta+\cos\gamma=:f(\alpha,\beta,\gamma),\\
16y^2 -1&=& 2\cos\alpha \cos\beta
\cos\gamma-\cos^2\alpha-\cos^2\beta-\cos^2\gamma=:g(\alpha,\beta,\gamma).
\end{eqnarray*}
\item(1) Note that $[f_{\min},f_{\max}]=\Br{-\frac32,3}$. Indeed,
\begin{itemize}
\item $\set{\bsr_1,\bsr_2,\bsr_3}$ is linearly dependent if and only
if either one of the following statements holds true:
$\alpha=\beta+\gamma$ or $\beta=\alpha+\gamma$ or
$\gamma=\alpha+\beta$ or $\alpha+\beta+\gamma=2\pi$, where
$(\alpha,\beta,\gamma)\in[0,\pi]^3$.
\item $\set{\bsr_1,\bsr_2,\bsr_3}$ is linearly independent if and only
if the following statements hold true:
\begin{eqnarray*}
\alpha,\beta,\gamma\in(0,\pi):
\alpha<\beta+\gamma,\beta<\alpha+\gamma,\gamma<\alpha+\beta,\alpha+\beta+\gamma<2\pi
.\end{eqnarray*}
\end{itemize}
The gradient of $f$ is given by $\nabla f=(\partial_\alpha
f,\partial_\beta f,\partial_\gamma f) =
-(\sin\alpha,\sin\beta,\sin\gamma)$. Clearly $f$ does not have
stationary points in $(0,\pi)^3$. We claim that
\begin{center}
$f$ attains extremal values only when $\set{\bsr_1,\bsr_2,\bsr_3}$
is linearly dependent.
\end{center}
Without loss of generality, we assume that $\gamma=\alpha+\beta$.
Thus, $f(\alpha,\beta,\gamma) =
\cos\alpha+\cos\beta+\cos(\alpha+\beta)$. Based on this observation,
$f_{\max}=3$ if $\alpha=\beta=\gamma=0$. Let $\tilde
f(\alpha,\beta)=\cos\alpha+\cos\beta+\cos(\alpha+\beta)$. Its
gradient is given by
$$
\nabla \tilde f=-(\sin\alpha+\sin(\alpha+\beta),
\sin\beta+\sin(\alpha+\beta))=0,
$$
implying that $\sin(\alpha+\beta)=-\sin\alpha=-\sin\beta$. Solving
this equation, we get that $\alpha=\beta=0$ or
$\alpha=\beta=\frac23\pi$. Therefore, $f_{\min}=-\frac32$.
\item (2 and 3) In order to find the maximal value of $y$, this is equivalent to find the maximal value of
$16y^2-1=g(\alpha,\beta,\gamma)$. Instead of finding the maximal
value of $g(\alpha,\beta,\gamma)$, we rewrite
$g(\alpha,\beta,\gamma)$ as
$G(t_1,t_2,t_3):=2t_1t_2t_3-t^2_1-t^2_2-t^2_3$, where
$(t_1,t_2,t_3)=(\cos\alpha,\cos\beta,\cos\gamma)$. Let us optimize
this objective function $G(t_1,t_2,t_3)$ subject to the constraint
$t_1+t_2+t_3=4x-1$. Construct Lagrange multiplier
$$
L(t_1,t_2,t_3,\lambda):=2t_1t_2t_3-t^2_1-t^2_2-t^2_3 +
\lambda(t_1+t_2+t_3-4x+1).
$$
Its gradient is
$$
\nabla L = (\lambda -2 t_1+2 t_2 t_3, \lambda -2 t_2+2 t_1 t_3,
\lambda -2 t_3+2 t_1 t_2,t_1+t_2+t_3-4x+1)=0,
$$
which gives rise to the following solutions:
\begin{eqnarray*}
(t_1,t_2,t_3,\lambda)&=&\Pa{\frac{4x-1}3,\frac{4x-1}3,\frac{4x-1}3,-\frac89(4 x^2-5 x+1)}\\
(t_1,t_2,t_3,\lambda)&=&(-1,-1,4x+1,8x),\\
(t_1,t_2,t_3,\lambda)&=&(-1,4x+1,-1,8x),\\
(t_1,t_2,t_3,\lambda)&=&(4x+1,-1,-1,8x).
\end{eqnarray*}
After calculations, we get that $G$ attains its maximum at
$(t_1,t_2,t_3)=(\frac{4x-1}3,\frac{4x-1}3,\frac{4x-1}3)$, i.e.,
$$
G_{\max}=1+2\Pa{\frac{4x-1}3}^3-3\Pa{\frac{4x-1}3}^2
$$
implies that
$$
16y^2-1\leqslant 2\Pa{\frac{4x-1}3}^3-3\Pa{\frac{4x-1}3}^2
$$
which can be rewritten as
$$
x^2+y^2\leqslant
1+2\Pa{\frac{4x-1}3}^3-3\Pa{\frac{4x-1}3}^2+16x^2=\Pa{\frac{2x+1}3}^3.
$$
The equation of the boundary curve is $x^2+y^2=\Pa{\frac{2x+1}3}^3$.
In summary,
$\max\abs{y}=\frac14\max_{\bsr_k:k=1,2,3}\det(\bsr_1,\bsr_2,\bsr_4)=\frac14=\tau_3$
when $x=\frac14$.
\end{proof}

In what follows, we calculate the largest imaginary part for
$\cB^\circ_4$. To that end, we establish the following interesting
inequality:
\begin{lem}\label{lem:intineq}
For any four unit vectors $\set{\bsr_k:k=1,2,3,4}$ in $\real^3$, it
holds that
\begin{eqnarray}\label{eq:C4}
-\frac{16}{3\sqrt{3}}\leqslant\det(\bsr_1+\bsr_2,\bsr_2+\bsr_3,\bsr_3+\bsr_4)\leqslant
\frac{16}{3\sqrt{3}}.
\end{eqnarray}
The upper bound can be attained, for instance, at the following
vectors:
$$
\bsr_1=\frac1{\sqrt{3}}\Pa{\sqrt{2},0,1},\bsr_2=\frac1{\sqrt{3}}\Pa{0,\sqrt{2},1},\bsr_3=\frac1{\sqrt{3}}\Pa{-\sqrt{2},0,1},\bsr_4=\frac1{\sqrt{3}}\Pa{0,-\sqrt{2},1}.
$$
\end{lem}

\begin{proof}
Due to the fact that
$$
\det(-\bsr_1-\bsr_2,-\bsr_2-\bsr_3,-\bsr_3-\bsr_4)=-\det(\bsr_1+\bsr_2,\bsr_2+\bsr_3,\bsr_3+\bsr_4)
$$
it suffices to consider the maximal value of the above determinant.
In fact, we see that
\begin{eqnarray*}
\det(\bsr_1+\bsr_2,\bsr_2+\bsr_3,\bsr_3+\bsr_4)&=&\det(\bsr_1,\bsr_2,\bsr_3)+\det(\bsr_1,\bsr_2,\bsr_4)+\det(\bsr_1,\bsr_3,\bsr_4)+\det(\bsr_2,\bsr_3,\bsr_4)\\
&=&(\bsr_1\times\bsr_2)\cdot\bsr_3+(\bsr_1\times\bsr_2)\cdot\bsr_4+\bsr_1\cdot(\bsr_3\times\bsr_4)+\bsr_2\cdot(\bsr_3\times\bsr_4).
\end{eqnarray*}
In order to get larger values of the above determinant, firstly,
each term on the right hand side must be nonnegative, i.e.,
$$
(\bsr_1\times\bsr_2)\cdot\bsr_3\geqslant0,(\bsr_1\times\bsr_2)\cdot\bsr_4\geqslant0,\bsr_1\cdot(\bsr_3\times\bsr_4)\geqslant0,\bsr_2\cdot(\bsr_3\times\bsr_4)\geqslant0.
$$
Denote by $H_{ij}$ the hyperplane determined by
$\bsr_i\times\bsr_j$, where $(i,j)\in\set{(1,2),(3,4)}$. Thus,
$\bsr_3,\bsr_4\in H^+_{12}$ and $\bsr_1,\bsr_2\in H^+_{34}$. From
this, we see that all four vectors $\bsr_1,\bsr_2,\bsr_3,\bsr_4$ are
in $H^+_{12}\cap H^+_{34}$. Furthermore, all four vectors are in the
convex cone whose vertex is just origin.
\begin{enumerate}
\item[(a)] Note that $\det(\bsr_i,\bsr_j,\bsr_k)\leqslant1$ by Proposition~\ref{prop:1}. By symmetry,
without loss of generality, we assume that
$\det(\bsr_1,\bsr_2,\bsr_3)=1$. This amounts to saying that
$\set{\bsr_1,\bsr_2,\bsr_3}$ is an orthonormal basis in a
right-handed system. We will see that the remaining determinant
would not arrive at one at the same time. Indeed, by symmetry, we
assume that $\det(\bsr_1,\bsr_2,\bsr_4)=1$, which holds true if and
only if $\set{\bsr_1,\bsr_2,\bsr_4}$ is also an orthonormal in a
right-handed system. This implies that $\bsr_4$ is orthogonal to the
plane spanned by $\bsr_1,\bsr_2$. Furthermore, $\bsr_4\propto\bsr_3$
because $\set{\bsr_1,\bsr_2,\bsr_3}$ is an orthonormal basis. So,
$\bsr_3=\pm\bsr_4$. But,
$\det(\bsr_1,\bsr_2,\bsr_3)=\det(\bsr_1,\bsr_2,\bsr_4)=1$, and we
see that $\bsr_3=\bsr_4$. In this case, we get that
\begin{eqnarray*}
\det(\bsr_1,\bsr_2,\bsr_3)+\det(\bsr_1,\bsr_2,\bsr_4)+\det(\bsr_1,\bsr_3,\bsr_4)+\det(\bsr_2,\bsr_3,\bsr_4)=2.
\end{eqnarray*}
From the above reasoning, we see that if there are two determinants
arriving at one, then the other two determinants must be vanished.
\item[(b)] Suppose there is only one determinant arriving at one, say,
$\det(\bsr_1,\bsr_2,\bsr_3)=1$, but the other three determinant do
not take one. Now, $\set{\bsr_1,\bsr_2,\bsr_3}$ form a right-handed
system, then without loss of generality, we assume that
$\set{\bsr_1,\bsr_2,\bsr_3}$ is a computational basis
$\set{\bse_1,\bse_2,\bse_3}$ and $\bsr_4$ varies arbitrarily. Then,
\begin{eqnarray*}
&&\det(\bse_1,\bse_2,\bse_3)+\det(\bse_1,\bse_2,\bsr_4)+\det(\bse_1,\bse_3,\bsr_4)+\det(\bse_2,\bse_3,\bsr_4)\\
&&=1+r_{41}-r_{42}+r_{43},
\end{eqnarray*}
where $\bsr_4=(r_{41},r_{42},r_{43})$. Now, optimize the objection
function of arguments $\bsr_4$ with the constraints
$\abs{\bsr_4}=1$. Construct Lagrange multiplier function
$$
L(\bsr_4,\lambda)=1+r_{41}-r_{42}+r_{43}+\lambda\Pa{\sum^3_{k=1}r^2_{4k}-1}.
$$
Then, its gradient is given by
$$
\nabla L=\Pa{2\lambda r_{41}+1,2\lambda r_{42}-1,2\lambda
r_{43}+1,\sum^3_{k=1}r^2_{4k}-1},
$$
which is vanished if and only if
$(\bsr_4,\lambda)=\pm\Pa{\frac1{\sqrt{3}},-\frac1{\sqrt{3}},\frac1{\sqrt{3}},-\frac{\sqrt{3}}2}$.
Note that
$(\bsr_1\times\bsr_2)\cdot\bsr_4=\bse_3\cdot\bsr_4\geqslant0$, which
means $r_{43}\geqslant0$. Therefore,
$(\bsr_4,\lambda)=\Pa{\frac1{\sqrt{3}},-\frac1{\sqrt{3}},\frac1{\sqrt{3}},-\frac{\sqrt{3}}2}$
implies that the sum of four determinants can arrive at a larger
value:
\begin{eqnarray*}
\det(\bse_1,\bse_2,\bse_3)+\det(\bse_1,\bse_2,\bsr_4)+\det(\bse_1,\bse_3,\bsr_4)+\det(\bse_2,\bse_3,\bsr_4)=1+\sqrt{3}>2.
\end{eqnarray*}
where $\det(\bsr_1,\bsr_2,\bsr_3)=1$ and
$\det(\bsr_1,\bsr_2,\bsr_4)=\det(\bsr_1,\bsr_3,\bsr_4)=\det(\bsr_2,\bsr_3,\bsr_4)=\frac1{\sqrt{3}}$.
\item[(c)] If there is no determinant arriving at one, then from the
above reasoning of symmetric argument, we can see that the sum of
four determinants will get larger in value when
$\det(\bsr_1,\bsr_2,\bsr_3)=\det(\bsr_1,\bsr_2,\bsr_4)=\det(\bsr_1,\bsr_3,\bsr_4)=\det(\bsr_2,\bsr_3,\bsr_4)$
get larger in value. Denote $V(\bsr_i,\bsr_j,\bsr_k)$ as the volume
of the triangle pyramid with three edges $(\bsr_i,\bsr_j,\bsr_k)$ at
a common vertex origin, where $1\leqslant i<j<k\leqslant 4$. We also
let $V(\bsr_1,\bsr_2,\bsr_3,\bsr_4)$ be the volume of the
quadrilateral pyramid with four edges
$\set{\bsr_1,\bsr_2,\bsr_3,\bsr_4}$. We get that
\begin{eqnarray*}
\det(\bsr_1+\bsr_2,\bsr_2+\bsr_3,\bsr_3+\bsr_4)= 6\sum_{1\leqslant
i<j<k\leqslant
4}V(\bsr_i,\bsr_j,\bsr_k)=12V(\bsr_1,\bsr_2,\bsr_3,\bsr_4),
\end{eqnarray*}
where
$V(\bsr_1,\bsr_2,\bsr_3,\bsr_4)=V(\bsr_1,\bsr_2,\bsr_3)+V(\bsr_1,\bsr_3,\bsr_4)=V(\bsr_2,\bsr_3,\bsr_4)+V(\bsr_1,\bsr_2,\bsr_4)$.
In what follows, we find the maximal volume $V_{\max}$ of
$V(\bsr_1,\bsr_2,\bsr_3,\bsr_4)$. For the fixed height $h\in(0,1)$
of the quadrilateral pyramid with four edges
$\set{\bsr_1,\bsr_2,\bsr_3,\bsr_4}$, when the base area achieves its
maximum, then the volume attains its maximum, thus at the height
$h$, we get that the local maximal volume is given by $\frac13 \cdot
2(1-h^2)\cdot h$ for $h\in(0,1)$. Now, the function
$\frac23(1-h^2)h$, achieves its maximum $\frac4{9\sqrt{3}}$ at
$h=\frac1{\sqrt{3}}$. Therefore, $V_{\max}=\frac4{9\sqrt{3}}$, and
thus
\begin{eqnarray*}
\det(\bsr_1+\bsr_2,\bsr_2+\bsr_3,\bsr_3+\bsr_4)\leqslant
\tfrac{16}{3\sqrt{3}}.
\end{eqnarray*}
\end{enumerate}
Attainability of the upper bound is apparent for
$$
\bsr_1=\Pa{\sqrt{\tfrac23},0,\sqrt{\tfrac13}},\bsr_2=\Pa{0,\sqrt{\tfrac23},\sqrt{\tfrac13}},\bsr_3=\Pa{-\sqrt{\tfrac23},0,\sqrt{\tfrac13}},\bsr_4=\Pa{0,-\sqrt{\tfrac23},\sqrt{\tfrac13}}.
$$
This completes the proof.
\end{proof}

\begin{prop}
It holds that
\begin{enumerate}
\item[(1)] The real part of 4th order Bargmann invariant
$x=\mathrm{Re}\Tr{\psi_1\psi_2\psi_3\psi_4}$ ranges over
$[-\frac14,1]$.
\item[(2)]  The imaginary part of 4th order Bargmann invariant
$y=\mathrm{Im}\Tr{\psi_1\psi_2\psi_3\psi_4}$ ranges over
$[-\frac2{3\sqrt{3}},\frac2{3\sqrt{3}}]$.
\end{enumerate}
\end{prop}

\begin{proof}
Now for
$t_1=\Inner{\bsr_2}{\bsr_3},t_2=\Inner{\bsr_1}{\bsr_3},t_3=\Inner{\bsr_1}{\bsr_2}$
and
$t_4=\Inner{\bsr_3}{\bsr_4},t_5=\Inner{\bsr_2}{\bsr_4},t_6=\Inner{\bsr_1}{\bsr_4}$,
we see that Eq.~\eqref{eq:A2} becomes
\begin{eqnarray*}
8x&=&\sum^6_{k=1}t_k + t_1t_6 - t_2t_5 + t_3t_4 = (1+t_1)(1+t_6)+(1+t_3)(1+t_4)-(1-t_2)(1-t_5),\\
8y &=& \det(\bsr_1+\bsr_2,\bsr_2+\bsr_3,\bsr_3+\bsr_4).
\end{eqnarray*}
Note that $\abs{t_k}\leqslant1$ for all $k$. Then, all
$0\leqslant1\pm t_k\leqslant2$, and thus
\begin{eqnarray*}
8x=(1+t_1)(1+t_6)+(1+t_3)(1+t_4)-(1-t_2)(1-t_5)\leqslant
2\times2+2\times2-(1-t_2)(1-t_5)\leqslant8.
\end{eqnarray*}
Still, also from the right-hand side expression of $8x$, in order to
attain its minimal value, it must be the case where the non-negative
factor $(1-t_2)(1-t_5)$ attains its maximal value $4$ since
$(1+t_1)(1+t_6)\geqslant0,(1+t_3)(1+t_4)\geqslant0$, and
$(1-t_2)(1-t_5)\leqslant4$. Now, $(1-t_2)(1-t_5)=4$ if and only if
$t_2=t_5=-1$, which is true if and only if
$\Inner{\bsr_1}{\bsr_3}=\Inner{\bsr_2}{\bsr_4}=-1$, or equivalently
$\bsr_1+\bsr_3=0=\bsr_2+\bsr_4$. Now, under this restricted
conditions $\bsr_1+\bsr_3=0=\bsr_2+\bsr_4$, we consider the optimal
problem
\begin{eqnarray*}
8x&=&(1+t_1)(1+t_6)+(1+t_3)(1+t_4)-(1-t_2)(1-t_5)\\
&\geqslant&(1+\Inner{\bsr_2}{\bsr_3})(1+\Inner{\bsr_1}{\bsr_4})+(1+\Inner{\bsr_1}{\bsr_2})(1+\Inner{\bsr_3}{\bsr_4})-4\\
&=&(1-\Inner{\bsr_1}{\bsr_2})^2+(1+\Inner{\bsr_1}{\bsr_2})^2-4 =
2\Inner{\bsr_1}{\bsr_2}^2-2\geqslant-2.
\end{eqnarray*}
Therefore, $8x\in[-2,8]$, $8x$ attains its maximal value $8$ when
$\bsr_1=\bsr_2=\bsr_3=\bsr_4$, and $8x$ attains its minimal value
$-2$ when $\bsr_1+\bsr_3=0=\bsr_2+\bsr_4$ and $\bsr_1\perp\bsr_4$.
Now, we can identity $\max\abs{y}
=\frac18\max_{\bsr_k:k=1,2,3,4}\det(\bsr_1+\bsr_2,\bsr_2+\bsr_3,\bsr_3+\bsr_4)$
by \eqref{eq:A2}. Using Lemma~\ref{lem:intineq}, we get immediately
that
\begin{eqnarray*}
\max\abs{y}
=\frac18\max_{\bsr_k:k=1,2,3,4}\det(\bsr_1+\bsr_2,\bsr_2+\bsr_3,\bsr_3+\bsr_4)
= \frac18 \frac{16}{3\sqrt{3}}=\frac2{3\sqrt{3}}=\tau_4.
\end{eqnarray*}
This completes the proof.
\end{proof}

\end{widetext}
\end{document}